\documentclass[11pt,reqno]{amsart}
\usepackage{amssymb}
\usepackage{stmaryrd}
\usepackage{physics}
\usepackage{braket}
\usepackage{dsfont}
\usepackage{graphicx}
\usepackage[table,xcdraw]{xcolor}
\usepackage{enumerate}
\usepackage[pagebackref, colorlinks = true, linkcolor = blue, urlcolor  = blue, citecolor = red]{hyperref}
\usepackage[margin=2.5cm]{geometry}
\usepackage{enumitem}

\usepackage{sudoku}

\usepackage{algorithm}
\usepackage[noend]{algpseudocode}
\makeatletter
\def\BState{\State\hskip-\ALG@thistlm}
\makeatother

\renewcommand{\epsilon}{\varepsilon}
\renewcommand{\phi}{\varphi}

\newtheorem{theorem}{Theorem}[section]
\newtheorem{definition}[theorem]{Definition}
\newtheorem*{definition*}{Definition}
\newtheorem{proposition}[theorem]{Proposition}

\newtheorem{lemma}[theorem]{Lemma}
\newtheorem{remark}[theorem]{Remark}
\newtheorem{example}[theorem]{Example}

\newtheorem{conjecture}[theorem]{Conjecture}
\newtheorem*{conjecture*}{Conjecture}

\newcommand{\QG}{\textsf{QG}}
\newcommand{\QS}{\textsf{QS}}
\newcommand{\QEG}{\textsf{QEG}}
\newcommand{\QUG}{\textsf{QUG}}
\newcommand{\CG}{\textsf{CG}}
\newcommand{\CS}{\textsf{CS}}
\newcommand{\CEG}{\textsf{CEG}}
\newcommand{\CUG}{\textsf{CUG}}

\usepackage{tikz}
\usetikzlibrary{positioning,chains,fit,shapes,calc}

\begin{document}

\definecolor{myblue}{RGB}{80,80,160}
\definecolor{mygreen}{RGB}{80,160,80}

\author{Ion Nechita}
\email{nechita@irsamc.ups-tlse.fr}
\address{Laboratoire de Physique Th\'eorique, Universit\'e de Toulouse, CNRS, UPS, France}

\author{Jordi Pillet}
\email{jordipillet@gmail.com}
\address{Universit\'e de Bourgogne, Dijon, France}

\title{SudoQ --- a quantum variant of the popular game}

\begin{abstract}
    We introduce \emph{SudoQ}, a quantum version of the classical game Sudoku. Allowing the entries of the grid to be (non-commutative) projections instead of integers, the solution set of SudoQ puzzles can be much larger than in the classical (commutative) setting. We introduce and analyze a randomized algorithm for computing solutions of SudoQ puzzles. Finally, we state two important conjectures relating the quantum and the classical solutions of SudoQ puzzles, corroborated by analytical and numerical evidence.
\end{abstract}

\date{\today}

\maketitle

\tableofcontents

\section{Introduction}

The \emph{Sudoku puzzle} has become nowadays one of the most popular pen-and-paper solitaire games. Its origin can be traced back to 1892, in the pages of the French monarchist daily newspaper ``Le Si\`ecle'', where a very similar puzzle was proposed to the readers. The Sudoku game can be seen as an extension of the older \emph{Latin square} game, already introduced by  Euler in the 18th century. The rules of the latter game are simple: you have to fill in a \( n \times n \) grid such that each row and each column contains a permutation of the elements of \( \{1,2,...,n\} \). A Sudoku puzzle is a \( n^2 \times n^2 \) Latin square with the supplementary constraints that each \( n \times n \) sub-square must also contain a permutation of \( \{1,2,...,n^2\} \); the most popular version of the game, usually found in newspapers, assumes \(n=3 \), presenting the puzzle as a $9 \times 9$ partially filled grid. The Sudoku puzzle has elicited numerous mathematical results, most of them dealing with the enumeration of puzzles having different properties. Importantly, it has been proven that the minimum number of clues (non empty cells) in any $9 \times 9$ proper Sudoku puzzle is 17 \cite{mcguire2014there}. 

In \cite{musto2016quantum}, the authors introduced \emph{quantum Latin squares}, which are non-commutative (or quantum) extensions of \( n \times n \) Latin squares. The elements of a quantum Latin square are now vectors in a finite dimensional, complex Hilbert space, with the rule being that the vectors in each row and each column must form an orthonormal basis (ONB) of the Hilbert space. In other words, the following \emph{quantization procedure} is used: 
\begin{align*}
\text{ element } a_{ij} \in \{1, 2, \ldots, n\} & \leadsto \text{ vector } x_{ij} \in \mathbb C^n \\
\text{ rows / cols are permutations of } \{1, 2, \ldots, n\} & \leadsto \text{ vectors in rows / cols form ONBs of } \mathbb C^n.
\end{align*}

One can restate the above properties in terms of the (unit-rank) orthogonal projections on the vectors. In this language, the elements of a square are projections $p_{ij} = \ketbra{x_{ij}}{x_{ij}}$ with the property that the row and the column sums are all equal to the identity. Such matrices, known as \emph{magic unitaries}, appear in the representation theory of the quantum permutation group $S_n^+$ \cite{wang1998quantum,banica2017flat,banica2019free}. Let us also mention two recent connections between the theory of quantum Latin squares and two problems in quantum information theory: Stinespring dilations \cite{benoist2017bipartite} and quantum homomorphisms of graphs \cite{atserias2019quantum}. 

\medskip

In this paper, we introduce the puzzle \emph{SudoQ}, a quantum generalization of the Sudoku game. The main idea is to impose to a quantum Latin square the supplementary constraint that each sub-square must also form an orthonormal basis of \( \mathbb{C}^n \). 

\begin{definition*}
A \emph{SudoQ square} of size $n^2$ is a $n^2 \times n^2$ matrix of vectors in $\mathbb C^{n^2}$ with the property that the $n^2$ vectors in each row, column, or $n \times n$ sub-square, form an orthonormal basis of $\mathbb C^{n^2}.$
\end{definition*}

The algebraic structure of sets of projections satisfying relations as above have been studied recently in relation to non-local games. The definition above can be seen as a restricted version of the  free hypergraph \( C^{*} \) algebras introduced in \cite{fritz2020curious}, in which one encodes the constraints as hyperedges. Given a hypergraph \(H=(V,E)\), the associated \emph{free hypergraph \( C^{*} \)-algebra} is the following finitely presented \( C^{*} \)-algebra:
\begin{equation*}
    C^{*}(H)=\bigg\langle (p_v)_{v \in V} \quad | \quad p_v^2=p_v=p_v^{*},\quad  \sum_{v \in e} p_v = \mathrm I,  \quad e \in E \bigg\rangle.
\end{equation*}
A first example is the one corresponding to the \emph{rook's hypergraph}. Let us imagine an empty chessboard, the vertices of the hypergraph correspond to all the squares, whereas the two hyperedges intersecting on a given square contain all the position accessible for a rook on this square (see Figure \ref{fig:hypergraph}, left panel). In the case of rook's hypergraph, this algebra corresponds exactly to the one associated with the quantum permutation group $S_n^+$, and the conditions are precisely the ones for a quantum Latin square. We consider in this paper the Sudoku hypergraph, which is naturally defined by adding to the rook hypergraph the sub-square constraints (see Figure \ref{fig:hypergraph}, right panel). 

\begin{figure}[ht]
  \centering
  \includegraphics[width=0.3\textwidth]{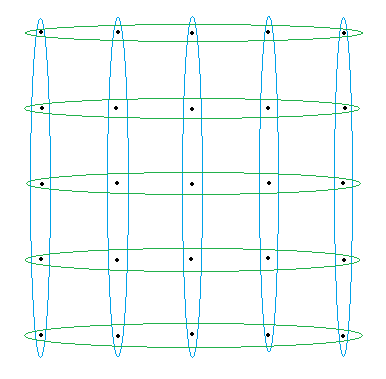} \qquad\qquad\qquad \includegraphics[width=0.3\textwidth]{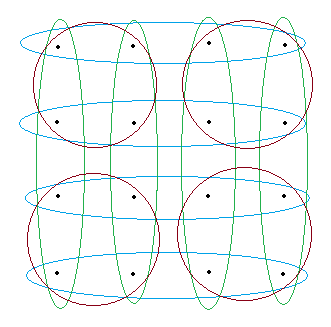}
  \caption{Rook's hypergraph of dimension $5 \times 5$ (left) and Sudoku hypergraph of dimension $2^2 \times 2^2$ (right).}
  \label{fig:hypergraph}
\end{figure}

After defining SudoQ squares, we introduce \emph{quantum grids}, which are just matrices of vectors, having possibly zero entries (to denote a missing element in a Sudoku puzzle). We consider the \emph{filling partial order} on the set of quantum grids, allowing us to analyze the solution space of a SudoQ grid. We analyze the relation between \emph{classical} and \emph{purely quantum} solutions of SudoQ grids. We call a grid classical if its entries are either zero or elements of the canonical basis of \( \mathbb{C}^{n^2} \) (usually designated by \( \{\ket{1},\ket{2},...,\ket{n^2} \} \)). The relation between the classical and quantum solution sets is the object of two main conjectures we can informally state as follows.

\begin{conjecture*}
    An unsolvable Sudoku puzzle does not have any quantum solutions. Similarly, a proper Sudoku puzzle (having a unique classical solution) has no purely quantum solutions.
\end{conjecture*}

We also introduce a \emph{SudoQ solver}, an algorithm for finding (approximate) solutions of SudoQ puzzles. The algorithm is based on a Sinkhorn-like alternating normalization method. Initially developed for bistochastic matrices \cite{sinkhorn1964relationship,sinkhorn1967concerning}, these algorithms have been recently generalized to the quantum (non-commutative) setting in different ways \cite{gurvits2004classical,benoist2017bipartite, banica2017flat,burgisser2018alternating,burgisser2018efficient,nechita2019sinkhorn}. The algorithm we present is close to the ones in \cite{banica2017flat} and \cite{nechita2019sinkhorn} (used, respectively, to generate random magic unitaries and random quantum symmetries of graphs), with the particularity that a subset of the entries of the grid are being kept fixed. We analyze numerically the algorithm and provide evidence for the main conjectures. We also discuss some applications to error correcting codes.

This paper is organised as follows:
\begin{itemize}
    \item In Sections \ref{sec:grids}, \ref{sec:classical}, \ref{sec:sudoq-solutions} we give some theoretical definitions and results on the spaces of SudoQ grids and solutions. The two main conjectures regarding SudoQ solutions are stated in Section \ref{sec:sudoq-solutions}.
    \item In Section \ref{sec:sudoq-algorithm} we describe a Sinkhorn-like algorithm for computing solutions of SudoQ puzzles. Numerical results both on the efficiency of our algorithm and supporting the two conjectures on the number of SudoQ solutions from Section \ref{sec:sudoq-solutions} are given in Section \ref{sec:experiments}.
    \item Finally, in Section \ref{sec:mixed} we consider a more general setting called ``mixed SudoQ'' and in Section \ref{sec:application} we present a quantum version of Sudoku code as an application of the previous results.
\end{itemize}

\section{Quantum grids, squares, and the SudoQ game}\label{sec:grids}

In this section we introduce the main algebraic objects we are going to study, the quantum grids and squares. Following \cite{musto2016quantum}, we are going to work in the linear algebraic framework (matrices of finite dimensional projections) rather than operator algebraic one from \cite{fritz2020curious}. 

Constraint lists, defined below, encode the different requirements for a matrix of rank-one projections to be a quantum Latin square or a quantum Sudoku square, and correspond to the hyperedges from Figure \ref{fig:hypergraph}.

\begin{definition}
	A \emph{constraint list} is a sequence $\mathcal C = (\mathcal C_n)_{n \geq 1}$, where $\mathcal C_n$ is a set of $n$-subsets of $[n] \times [n]$. 
\end{definition}

\begin{example}
	The constraints for a quantum Latin square are that every row and column of a matrix of projections contain an orthonormal basis of the Hilbert space. We have thus
	$$\mathcal L_n = \{ \{(i, 1), (i, 2), \ldots (i,n)\} \, : \, i \in [n] \} \sqcup \{ \{(1, j), (2, j), \ldots (n, j)\} \, : \, j \in [n] \}.$$
\end{example}

\begin{example}\label{ex:Sudoku-constraints}
	For Sudoku squares, the size must be $n^2$ and the constraints are those for a quantum Latin square, plus the ones corresponding to the $n^2$ squares of size $n \times n$:
	\begin{equation}\label{eq:sudoku-constraints}
	\mathcal S_{n^2} = \mathcal L_{n^2} \sqcup \{ \{((I,i), (J,j)) \, : i,j \in [n]\} \, : \, I,J \in [n]\}.
	\end{equation}
	Above, elements in $[n^2]$ are identified with pairs $(K,k)$, $K,k \in [n]$. 
\end{example}
\begin{remark}
    It is of course possible to have Sudoku squares of size $mn$, seen as $m \times n$ grids of $n \times m$ blocks. We shall not consider such type of puzzles (and constraints) in this paper. 
\end{remark}

First, we introduce the notion of \emph{quantum grids}, which are simply matrices of zero- or unit-rank projections. 

\begin{definition}\label{def:quantum-grid}
	A \emph{quantum grid} of size $n$ is a an element $P \in M_n(M_n^{sa}(\mathbb C))$ with the property that for all $i,j \in [n]$, $P_{ij}$ is either 0 or a rank-one orthogonal projection in $M_n(\mathbb C)$. Non-zero elements $P_{ij}$ are called \emph{clues} or \emph{givens}. The set of quantum grids of size $n$ is denoted by $\QG_n$. 
\end{definition}

A quantum grid should be thought of a partial matrix, where the 0 entries correspond to unknown matrix elements. Note that we do not impose any constraint on the non-zero elements of a quantum grid. Constraints appear in the notion of a \emph{quantum square}, defined below. 

\begin{definition}
	Let $\mathcal C$ be a constraint list. A \emph{quantum square} of type $\mathcal C$ and dimension $n$ is an element $P \in M_n(M_n^{sa}(\mathbb C))$, where 
	\begin{itemize}
		\item for all $i,j \in [n]$, $P_{ij}$ is a rank-one orthogonal projection in $M_n(\mathbb C)$
		\item for all $\{(i_1,j_1), \ldots, (i_n,j_n)\} \in \mathcal C_n$, 
		$$\sum_{s=1}^n P_{i_s,j_s} = I_n.$$
	\end{itemize}
	The set of $\mathcal C$-quantum squares of size $n$ is denoted by $\QS^{\mathcal C}_n \subseteq \QG_n$. 
\end{definition}

\begin{lemma}
	The second property in the definition above can be restated as follows: for any constraint $\{(i_1,j_1), \ldots, (i_n,j_n)\} \in \mathcal C_n$, the vectors $x_1, x_2, \ldots, x_n$ on which $P_{i_s, j_s}$ project, form an orthonormal basis of $\mathbb C^n$.
\end{lemma}

Several remarks are in order now. Firstly, note that we prefer to work with orthogonal projections instead of vectors, in order to render the structure of the squares more linear and to replace the orthonormal basis constraint with a linear one. Secondly, note that the difference between the definition above and \cite[Definition 2.1]{fritz2020curious} is the two restrictions we add. For a quantum square $P \in \QS_n$, we require that
\begin{itemize}
    \item the elements $P_{ij}$ are finite dimensional matrices: $P_{ij} \in M_n^{sa}(\mathbb C)$
    \item the projections $P_{ij}$ have unit rank.
\end{itemize}
These restrictions render the existence problem for quantum squares of a given type very different from the one in \cite{fritz2020curious} (see also the discussion in \cite[Section 2]{fritz2020curious}). We shall not study the existence problem in this work, because we are going to consider only the Sudoku constraint list from Example \ref{ex:Sudoku-constraints} for which solutions clearly exist, see Proposition \ref{prop:SudoQ-existence}.  

\begin{definition}
    A \emph{SudoQ square} is a quantum square satisfying the Sudoku constraint list from Example \ref{ex:Sudoku-constraints}. 
\end{definition}

Note that the definition above and the one given in the introduction are the same, in virtue of the following simple lemma, the proof of which is left to the reader.  

\begin{lemma}
    Let $P_{i} = \ketbra{x_i}{x_i}$ be a $n$-tuple of unit rank self-adjoint projections from $M_n(\mathbb C)$. Then,
    $$\sum_{i=1}^n P_i = I_n \iff \{x_1, \ldots, x_n\} \text{ is an ONB of } \mathbb C^n.$$
\end{lemma}

\begin{proposition}\label{prop:SudoQ-existence}
For all integers $n \geq 1$, the set of SudoQ squares $\QS_{n^2}^{\mathcal S_n}$ is non-empty. 
\end{proposition}
\begin{proof}
Start from a classical Latin square $S$ if size $n$ (e.g.~the one which has $1, 2, \ldots, n$ on the first row, and the subsequent rows are circular permutations of the first one). Build the Sudoku square $G$ by shifting-and-copying the square $S$ (note the roles $a$ and $i$ play in the expression below):
$$G_{(i,a), (b,j)} = S_{i,j} + n(S_{a,b}-1).$$
The $n=3$ case is displayed below. 
\begin{sudoku}
|1|2|3|4|5|6|7|8|9|.
|7|8|9|1|2|3|4|5|6|.
|4|5|6|7|8|9|1|2|3|.
|3|1|2|6|4|5|9|7|8|.
|9|7|8|3|1|2|6|4|5|.
|6|4|5|9|7|8|3|1|2|.
|2|3|1|5|6|4|8|9|7|.
|8|9|7|2|3|1|5|6|4|.
|5|6|4|8|9|7|2|3|1|.
\end{sudoku}
The SudoQ square $Q$ is obtained by adding kets around the elements of $G$: $Q_{x,y} = \ketbra{G_{x,y}}{G_{x,y}}$ (see also Lemma \ref{lem:quantization}). 
\end{proof}

\begin{definition}
    Given two quantum grids $P, Q \in \QG_n$, we say that $Q$ is a \emph{filling} of $P$ (and we write $P \preceq Q$) if 
	$$P_{ij} \neq 0 \implies Q_{ij} = P_{ij}.$$
	In other words, the set of zero entries of $Q$ is smaller than the set of zero entries of $P$. 
\end{definition}

The filling relation is a partial order on the set of quantum grids, with the all-zero grid being the minimal element and quantum squares (i.e.~grids having no zero elements) being the maximal elements. 

We introduced next the set of solvable, resp.~uniquely solvable grids. 

\begin{definition} A quantum grid $P \in \QG_n$ is called (uniquely) $\mathcal C$-solvable if there exists (an unique) $Q \in \QS_n^{\mathcal C}$ with $P \preceq Q$. The set of $\mathcal C$-solvable (resp.~uniquely solvable) grids is denoted by $\QEG_n^{\mathcal C}$, resp.~$\QUG_n^{\mathcal C}$
\end{definition}

We have
$$\QS_n^{\mathcal C} \subseteq \QUG_n^{\mathcal C} \subseteq \QEG_n^{\mathcal C} \subseteq \QG_n.$$ 

The different notions introduced in this section are gathered in Table \ref{tab:notations-quantum}. 

\begin{table}[ht]
\begin{tabular}{|r|l|c|}
\hline
\rowcolor[HTML]{EFEFEF} 
Notation & Description & Example \\ \hline\hline
$\QG$ & \begin{tabular}[c]{@{}l@{}}\emph{Quantum grids} --- matrices of vectors where \\ entries are either zero or unit vectors.\end{tabular} & 
$\begin{bmatrix} 
\cdot & \cdot & \cdot & \cdot \\
\cdot & \cdot & \cdot & \cdot \\
\ket 4 & \ket + & \ket - & \ket 4 \\
\ket - & \ket 3 & \ket 4 & \ket +
\end{bmatrix}$   \\ \hline
$\QEG$ & \begin{tabular}[c]{@{}l@{}}\emph{Quantum solvable grids} --- grids having \\ at least one solution. Entries can be empty.\end{tabular} & 
$\begin{bmatrix} 
\cdot & \cdot & \ket 3 & \ket 4 \\
\cdot & \cdot & \ket + & \ket - \\
\ket 4 & \ket + & \ket - & \ket 3 \\
\ket - & \ket 3 & \ket 4 & \ket +
\end{bmatrix}$   \\ \hline
$\QUG$ & \begin{tabular}[c]{@{}l@{}}\emph{Quantum uniquely solvable grids} --- grids having \\ a unique solution. Entries can be empty.\end{tabular} & 
$\begin{bmatrix} 
\cdot & \cdot & \ket 3 & \ket 4 \\
\ket 3 & \ket 4 & \ket + & \ket - \\
\ket 4 & \ket + & \ket - & \ket 3 \\
\ket - & \ket 3 & \ket 4 & \ket +
\end{bmatrix}$   \\ \hline
$\QS$ & \begin{tabular}[c]{@{}l@{}}\emph{Quantum Squares} --- completely filled out \\ SudoQ grids. Each entry is a unit vector.\end{tabular} &  $\begin{bmatrix} 
\ket + & \ket - & \ket 3 & \ket 4 \\
\ket 3 & \ket 4 & \ket + & \ket - \\
\ket 4 & \ket + & \ket - & \ket 3 \\
\ket - & \ket 3 & \ket 4 & \ket +
\end{bmatrix}$   \\ \hline
\end{tabular}
\caption{Table gathering the different notation used for \emph{quantum} grids and squares. We use the notation $\ket \pm := (\ket 1 \pm \ket 2)/\sqrt 2$, which is standard in quantum information theory.}
\label{tab:notations-quantum}
\end{table}

\begin{remark}
    Since most of the times we shall consider the Sudoku constraints from Example \ref{ex:Sudoku-constraints}, when using the notation $\QS, \QUG, \QEG$ without the superscript, we mean the sets associated to $\mathcal C = \mathcal S$.
\end{remark}

\section{Classical grids and squares}\label{sec:classical}

From now on, we fix a basis of $\mathbb C^n$ which we call the \emph{computational basis} and we denote it by $\{\ket{k} \}_{k \in n}$. In this section, we shall the notions defined previously, restricted to projections on vectors in the computational basis. 

\begin{definition}
	A quantum grid $P \in \QG_n$ is called \emph{classical} if, for all $i,j \in [n]$, the projection $P_{ij}$ is either 0 or one of the $\ketbra{k}{k}$, for $k \in [n]$. The set of classical grids is denoted by $\CG_n$. Similarly, we define
	\begin{itemize}
	    \item $\CS_n^{\mathcal C} := \QS_n^{\mathcal C} \cap \CG_n$, the set of classical squares. These are the usual (solved) squares
	    \item $\CEG_n^{\mathcal C}$, the set of classical grids which admit \emph{at least one classical solution}
	    \item $\CUG_n^{\mathcal C}$, the set of classical grids which admit a \emph{unique classical solution}. 
	\end{itemize}

\end{definition}
This definition allows us to think about a canonical quantization morphism which takes as input a grid valuated in \( \mathbb{N}\) and associates an element of $\CG_n$ as output. Intuitively it corresponds to just replace the natural number \(k \) in some cell by the corresponding ket \(\ket{k} \) in the computation basis. See Table \ref{tab:notations-classical}.

\begin{lemma}\label{lem:quantization}
	Well-posed classical grids are precisely the grids in $\CUG_n$: they admits a unique quantum solution in \(\CS_n \), which is the quantization of its classical solution. 
\end{lemma}
\begin{proof}
We start with a well-posed classical grid valuated in \( \mathbb{N} \). This grid can be put bijectively in correspondence with an element \(A\) of $\CG_n$ through the canonical quantization morphism
$$\{1,2, \ldots, n\} \ni k \mapsto \ket k \in \mathbb C^n.$$
Now we have two things to prove:
\begin{enumerate}[label=(\roman*)]
    \item There exists a square \(B\) in $\CS_n$ which solves \(A\).
    \item This solution is unique in $\CS_n$.
\end{enumerate}
To prove (i) we can first remark that we know a square \(B_c\) valuated in \( \mathbb{N} \) which solves the grid we start with in a classical sense. We apply the quantization morphism to \(B_c\); our claim is that this quantum square \(B\) solves \(A\). This is obviously true because each row and column of \(B\) is made of permutations of the canonical projections \( \ketbra{k}{k}\) with no more than one occurrence in each row, column, sub-square. Hence \(B \in \CS_n \). Moreover we also have \(A \preceq B \) because the set of zeros of \( B \) is empty (and the pre-filled projectors inside \(A\) are unchanged).The uniqueness of the solution (point (ii) above) follows from the fact that starting from \(A \) its solution \(B \) is obtained through a sequence of one to one correspondences. Indeed \(A \) via the inverse of the quantization morphism is sent to a well posed grid valuated in \( \mathbb{N} \), call it \(A_c \), which has a unique solution \(B_c \) ; then \(B_c \) is sent to \(B \) bijectively through the quantization morphism again.
\end{proof}

\begin{table}[ht]
\begin{tabular}{|r|l|c|}
\hline
\rowcolor[HTML]{EFEFEF} 
Notation & Description & Example \\ \hline\hline
$\CG$ & \begin{tabular}[c]{@{}l@{}}\emph{Classical grids} --- matrices of vectors where entries \\ are either zero or an element of the canonical basis.\end{tabular} & 
$\begin{bmatrix} 
\cdot & \cdot & \cdot & \cdot \\
\cdot & \cdot & \cdot & \cdot \\
4 & 4 & 4 & 4 \\
2 & 3 & 4 & 1
\end{bmatrix}$\\ \hline
$\CEG$ & \begin{tabular}[c]{@{}l@{}}\emph{Classical solvable grids} --- grids having \\ at least one classical solution. Entries can be empty.\end{tabular} & 
$\begin{bmatrix} 
\cdot & \cdot & \cdot & \cdot \\
\cdot & \cdot & \cdot & \cdot \\
4 & 1 & 2 & 3 \\
2 & 3 & 4 & 1
\end{bmatrix}$\\ \hline
$\CUG$ & \begin{tabular}[c]{@{}l@{}}\emph{Classical uniquely solvable grids} --- grids having \\ a unique classical solution. Entries can be empty.\end{tabular} & 
$\begin{bmatrix} 
\cdot & \cdot & 3 & 4 \\
\cdot & \cdot & 1 & 2 \\
4 & 1 & 2 & 3 \\
2 & 3 & 4 & 1
\end{bmatrix}$\\ \hline
$\CS$ & \begin{tabular}[c]{@{}l@{}}\emph{Classical Squares} --- completely filled out Sudoku grids.\\Each entry is an element of the canonical basis.\end{tabular} &  $\begin{bmatrix} 
1 & 2 & 3 & 4 \\
3 & 4 & 1 & 2 \\
4 & 1 & 2 & 3 \\
2 & 3 & 4 & 1
\end{bmatrix}$   \\ \hline
\end{tabular}
\caption{Table gathering the different notation used for \emph{classical} grids and squares.}
\label{tab:notations-classical}
\end{table}

\section{SudoQ solutions}\label{sec:sudoq-solutions}

We discuss in this section different questions regarding the SudoQ game and its relation to its classical counterpart. We shall consider classical grids having 0, 1, or $\geq 2$ classical solutions and discuss the relation between their classical and quantum solutions. 

The first natural question is whether impossible classical grids can be solved in the quantum world. 

\begin{conjecture}\label{conj:classical-1-quantum-1}
	Every quantumly solvable classical grid is classically solvable: 
	$$\QEG_{n^2} \cap \CG_{n^2} = \CEG_{n^2}.$$
\end{conjecture}
Above, the ``$\supseteq$'' inclusion is obviously true, hence we could rephrase the conjecture above using ``$\subseteq$'' instead of an equality.  

A different interesting question regarding the quantum generalization of Sudoku concerns the classical grids having exactly one classical solution (the newspaper grids). Such a grid admits the quantum solution corresponding to the classical one (obtained by associating to an integer the corresponding basis element). Do such grids admit more, purely quantum solutions?

\begin{conjecture}\label{conj:classical-1-quantum-more}
    Every classical grid having a unique classical solution does not have any extra quantum solutions
	$$\CUG_{n^2} \subseteq \QUG_{n^2} \cap \CG_{n^2}.$$
\end{conjecture}

We show in the two following examples that the situation is more complicated in the case where a classical grid admits $\geq 2$ classical solutions.

\begin{example}
	The exist non-uniquely solvable classical Sudoku grids which admit purely quantum solutions: 
	$$\begin{bmatrix}
	1 & 2 & 3 & 4 \\
	3 & 4 & \cdot & \cdot \\
	4 & 3 & \cdot & \cdot \\
	2 & 1 & 4 & 3
	\end{bmatrix}$$
	For any orthonormal basis $\{x,y\}$ of $\mathbb C \ket 1 \oplus \mathbb C \ket 2$, the following is a solution:
	$$\begin{bmatrix}
	x & y\\
	y & x
	\end{bmatrix}.$$
\end{example}

\begin{example}
	The exist non-uniquely solvable classical Sudoku grids which do not admit purely quantum solutions: 
\begin{sudoku}
	|1| | |4|5|6|7|8|9|.
	|4|5|6|7|8|9|1|2|3|.
	|7|8|9|1|2|3|4|5|6|.
	|2| | |5|6|4|8|9|7|.
	|5|6|4|8|9|7|2|3|1|.
	|8|9|7|2|3|1|5|6|4|.
	|3| | |6|4|5|9|7|8|.
	|6|4|5|9|7|8|3|1|2|.
	|9|7|8|3|1|2|6|4|5|.
\end{sudoku}
	
The grid above admits the following two classical solutions:
\begin{equation}\label{eq:2-sol-classical}
\begin{bmatrix}
2 & 3 \\
3 & 1 \\
1 & 2
\end{bmatrix} \quad \text{ and } \quad 
\begin{bmatrix}
3 & 2 \\
1 & 3 \\
2 & 1
\end{bmatrix}.
\end{equation}
However, this grid admits no purely quantum solution. Indeed, let $\{a_{ij}\}_{i\in[3], j \in [2]}$ be the missing values from the grid above. From the Sudoku constraints, it follows that 
$$\begin{bmatrix}
	a_{11} & a_{12} & \ketbra{1}{1} \\
	a_{21} & a_{22} & \ketbra{2}{2} \\
	a_{31} & a_{32} & \ketbra{3}{3}
\end{bmatrix}$$
is a quantum Latin square of size 3. But all $3 \times 3$ quantum Latin squares (or magic unitaries) are commutative, see \cite{wang1998quantum} or \cite[Section 2.2]{lupini2017nonlocal} for a short elementary proof. Hence, the rank-one projections $a_{ij}$ must belong to the computational basis, and thus the missing elements must be one of the two possible classical solutions from \eqref{eq:2-sol-classical}.
\end{example}

We consider now an example of an \emph{impossible Sudoku grid}, that is a classical grid which does not admit a classical solution. We shall prove that, although it does not have any quantum solutions either, the quantum square which minimizes the error (w.r.t.~the SudoQ property) yields an error term which is strictly smaller than the corresponding classical error. In other words, there exists a better approximate quantum solution than any classical one. 

Let us first define, for a quantum square, its distance with respect to the SudoQ constraint lists. See Section \ref{sec:sudoq-algorithm}, eq.~\eqref{eq:sinkhorn-error-algorithm} for the same definition used in the SudoQ algorithm, 

\begin{definition}
    Given a quantum grid $X \in \QG_{n^2}$, we define its \emph{SudoQ error} by
    $$\mathcal E(X) := \max_{c \in \mathcal S_{n^2}} \|\sum_{(i,j) \in c} \ketbra{x_{ij}}{x_{ij}} - I_{n^2}\|_2,$$
    where $\mathcal S_{n^2}$ is the SudoQ constraint list from Example \ref{ex:Sudoku-constraints}. 
    
    For a quantum grid $A \in \QG_{n^2}$, we define its \emph{SudoQ score} as
    $$\mathcal E_q(A) := \min_{Y \in \QS_{n^2}^{\mathcal S}, \, A \prec Y} \mathcal E(Y).$$

    For a classical grid $B \in \CG_{n^2}$, we define its \emph{Sudoku score} as
    $$\mathcal E_c(B) := \min_{Z \in \CS_{n^2}^{\mathcal S}, \, B \prec Z} \mathcal E(Z).$$
\end{definition}

Even if the two problems are not the same, the situation could be compared to the problem of finding for instance a graph bi-colouring for a graph made of a unique odd-cycle : we can replace colors by projectors and the bi-colouring condition by an orthogonality relation ; in the setting of non-local graph colouring games there exists quantum strategies which are better than classical strategies but we cannot find perfect strategies. 

Consider the following classical square, for $n=2$:
$$G = \begin{bmatrix}
	3 & \cdot & 1 & 4 \\
	\cdot & 4 & 3 & 2 \\
	1 & 2 & 4 & 3 \\
	4 & 3 & 2 & 1
	\end{bmatrix}$$
and denote by $x,y\in \mathbb C^4$ the missing vectors in the top, resp.~bottom row.

\begin{lemma}
    For the square $G$ above, the minimum error $\mathcal E_c(G)$ of a classical square is $\sqrt 2$, achieved, e.g., by $x=y=2$. The minimum quantum error is $\mathcal E_q(G) = 1$, achieved, e.g., by 
    $$x = \ket + = \frac{1}{\sqrt 2}(\ket 1 + \ket 2) \quad \text{and} \quad y = \ket - = \frac{1}{\sqrt 2}(\ket 1 - \ket 2).$$
\end{lemma}
\begin{proof}
It is clear that the grid $G$ is not classically solvable, and for any unsolvable grid the classical error is at least $\sqrt 2$. For the lower bound in the quantum case, the (quantum) error reads 
\begin{align*}\mathcal E(G) = \max\big\{ &\|\ketbra{x}{x} - \ketbra{1}{1}\|_2, \|\ketbra{x}{x} - \ketbra{2}{2}\|_2,\\ &\|\ketbra{y}{y} - \ketbra{1}{1}\|_2, \|\ketbra{y}{y} - \ketbra{2}{2}\|_2,\\ &\|\ketbra{x}{x} + \ketbra{y}{y} - I_2\|_2 \big\} .
\end{align*}
The five terms above correspond, respectively, to the SudoQ constraints of the second column, the first row, the second row, the first column, and the top-left square. Notice that
$$\|\ketbra{x}{x} - \ketbra{1}{1}\|_2^2 + \|\ketbra{x}{x} - \ketbra{2}{2}\|_2^2 = 2(1-|\langle 1, x\rangle|^2) + 2(1-|\langle 2, x\rangle|^2) \geq 2,$$
hence
$$\max\big\{ \|\ketbra{x}{x} - \ketbra{1}{1}\|_2^2, \|\ketbra{x}{x} - \ketbra{2}{2}\|_2^2\big\}\geq 1,$$
proving the claim. 
\end{proof}

\section{A Sinkhorn-like algorithm for SudoQ}\label{sec:sudoq-algorithm}

We describe in this section an algorithm, based on iterative scaling, which aims to produce a SudoQ square starting from an incomplete grid. The algorithm we present can be understand as a ``quantization'' of the one given in \cite{moon2009sinkhorn}. 

Let us first present in an informal manner our algorithm. Being based on Sinkhorn alternating normalization (see \cite{idel2016review} for a review), the main idea is to satisfy the constraints one after the other, in a cyclical order:
\begin{enumerate}
	\item Initialize the empty cells of a SudoQ grid with random vectors
	\item For each row: normalize the vectors corresponding to empty cells in the row in such a way that the row sum is the identity
	\item For each column: normalize the vectors corresponding to empty cells in the column in such a way that the column sum is the identity
	\item For each sub-square: normalize the vectors corresponding to empty cells in the sub-square in such a way that the sub-square sum is the identity
	\item Repeat the last three steps until all the constraints are approximately satisfied or the maximal number of steps has been reached. 
\end{enumerate}

We give now the full description of our Sinkhorn-based algorithm for solving SudoQ puzzles. The algorithm takes the following input data
\begin{itemize}
	\item $A$, a classical Sudoku grid, containing elements from $\{0,1,2, \ldots, n^2\}$, the 0 value corresponding to empty cells
	\item $\sigma$, the strength, a meta-parameter which can be chosen empirically and controls the intensity of the alternating normalizations
	\item $\epsilon$, the desired precision with which all the Sudoku constraints must be satisfied
	\item $I_{\max}$, the maximum number of iterations the algorithm will perform.
\end{itemize}

The algorithm will either return a SudoQ grid $x$, $\epsilon$-satisfying all the Sudoku constraints, or ``failure'' if no such square could be found before $I_{\max}$ Sinkhorn iterations. 

\begin{algorithm}
	\caption{Algorithm for solving SudoQ}
	\begin{algorithmic}[1]
		\Procedure{SolveSudoQ}{$A,\sigma,\epsilon, I_{\max}$}
		\State $x \gets \text{random $n^2 \times n^2 \times n^2$ complex Gaussian tensor}$\Comment{random initialization}
		\State $iter \gets 0$
		\Repeat
		\State $iter \gets iter + 1$
		\For{$c \in \textsc{ConstraintListSudoku}(n)$}
			\State $S \gets 0$
			\State $T \gets \mathds 1$
			\For{$(i,j) \in c$} 
				\If {$A_{ij} = 0$}
					\State $S \gets S +  \ketbra{x_{ij}}{x_{ij}}$ \Comment{empty cells}
				\Else
					\State $T \gets T- \ketbra{x_{ij}}{x_{ij}}$ \Comment{orth.~of filled cells}
				\EndIf
			\EndFor

			\State $R \gets \textsc{OptimalTransformation}(S,T)$ \Comment{transforms $S \to T$}
			\State $\tilde R \gets \sigma R + (1-\sigma)\mathds 1$ \Comment{control step size by interpolating}
			\For{$(i,j) \in c$} 
				\If {$A_{ij} = 0$}
					\State $x_{ij} \gets \tilde R x_{ij}$ \Comment{apply the transformation}
				\EndIf
			\EndFor
		\EndFor
		
		\Until{$\textsc{SinkhornError}(x) < \epsilon$ \textbf{or} $iter > I_{\max}$}\Comment{end conditions}
		\If {$\textsc{SinkhornError}(x) < \epsilon$}
		\State \Return{ $x$} \Comment{SudoQ solved}
		\Else
		\State \Return{$\text{failure}$} \Comment{too many iterations, at least one constraint not satisfied}
		\EndIf
		\EndProcedure
	\end{algorithmic}
	\label{alg:SudoQ}
\end{algorithm}

Algorithm \ref{alg:SudoQ} references three functions: $\textsc{ConstraintListSudoku}(n)$, $\textsc{OptimalTransformation}(S,T)$, and $\textsc{SinkhornError}(x)$. The function $\textsc{ConstraintListSudoku}$ returns a list corresponding to the $3n$ constraints of the Sudoku puzzle of size $n^2$: one for each row, column, and sub-square, see \eqref{eq:sudoku-constraints}. The function $\textsc{SinkhornError}(x)$ returns the largest error (computed in Frobenius norm) for the SudoQ square $x$ with respects to the Sudoku constraints: 
\begin{equation}\label{eq:sinkhorn-error-algorithm}
\textsc{SinkhornError}(x) = \max_{c \in \textsc{ConstraintListSudoku}(n)} \left\|\mathds 1_{n^2} - \sum_{(i,j) \in c} \ketbra{x_{ij}}{x_{ij}}\right\|.
\end{equation}

Finally, let us discuss in detail the function $\textsc{OptimalTransformation}(S,T)$, which is the key operation in our algorithm, both from a mathematical and from a computation complexity perspective. This function takes two projections $S,T$ of the same rank and computes the transformation $R$ satisfying $RSR^* = T$ which is closest to the identity on the support of $S$. We refer the reader interested in the details and the linear algebra behind this procedure to \cite[Lemma 3.4]{nechita2019sinkhorn}. Note we do not use precisely the mapping $R$ in our algorithm, but an interpolation $\tilde R$ between $R$ and the identity map, with $\sigma$ quantifying the contribution of $R$. This allows us to control the size of each step of the Sinkhorn alternating normalization procedure. 

\section{Numerical experiments}\label{sec:experiments}

We have tested our algorithm on most of the grids from \url{http://norvig.com/sudoku.html}. We report in this section the results obtained, commenting on their significance. We would like to point out, from the beginning, that the point of our algorithm is not to solve classical Sudoku puzzles, but to investigate whether non-classical (or purely quantum) solutions exist. Python code and the different grids analyzed here are available at \url{https://github.com/inechita/SudoQ}.

First, we have ran the Algorithm \ref{alg:SudoQ} multiple times on the following grid:  
\begin{sudoku}
	| | |3| |2| |6| | |.
	|9| | |3| |5| | |1|.
	| | |1|8| |6|4| | |.
	| | |8|1| |2|9| | |.
	|7| | | | | | | |8|.
	| | |6|7| |8|2| | |.
	| | |2|6| |9|5| | |.
	|8| | |2| |3| | |9|.
	| | |5| |1| |3| | |.		
\end{sudoku}
First, to determine the optimal value of the strength parameter $\sigma$, we have ran our balancing algorithm for $N = 1000$ times for each value of $\sigma \in \{s/10 \, : \, 1 \leq s \leq 10\}$. We report the results in Figure \ref{fig:results-norvig}.

\begin{figure}[ht]
	\centering
	\includegraphics[width=0.4\linewidth]{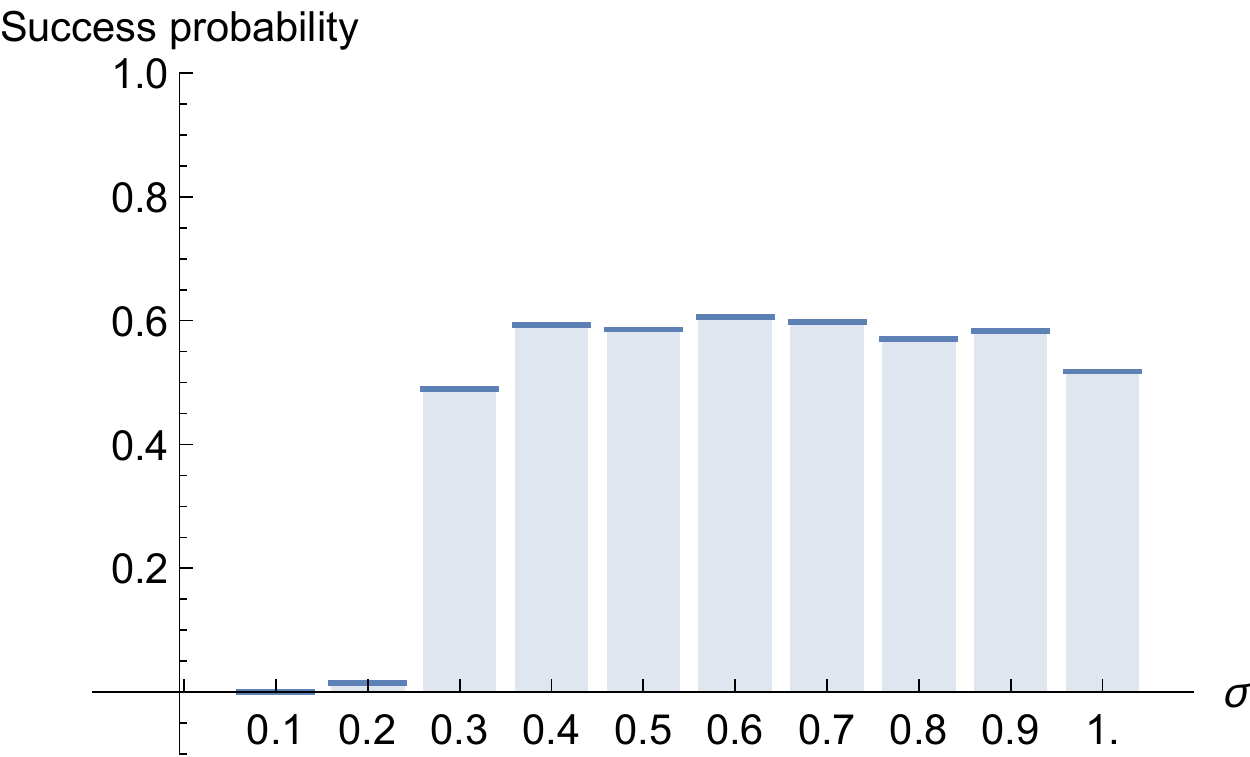} \qquad
	\includegraphics[width=0.4\linewidth]{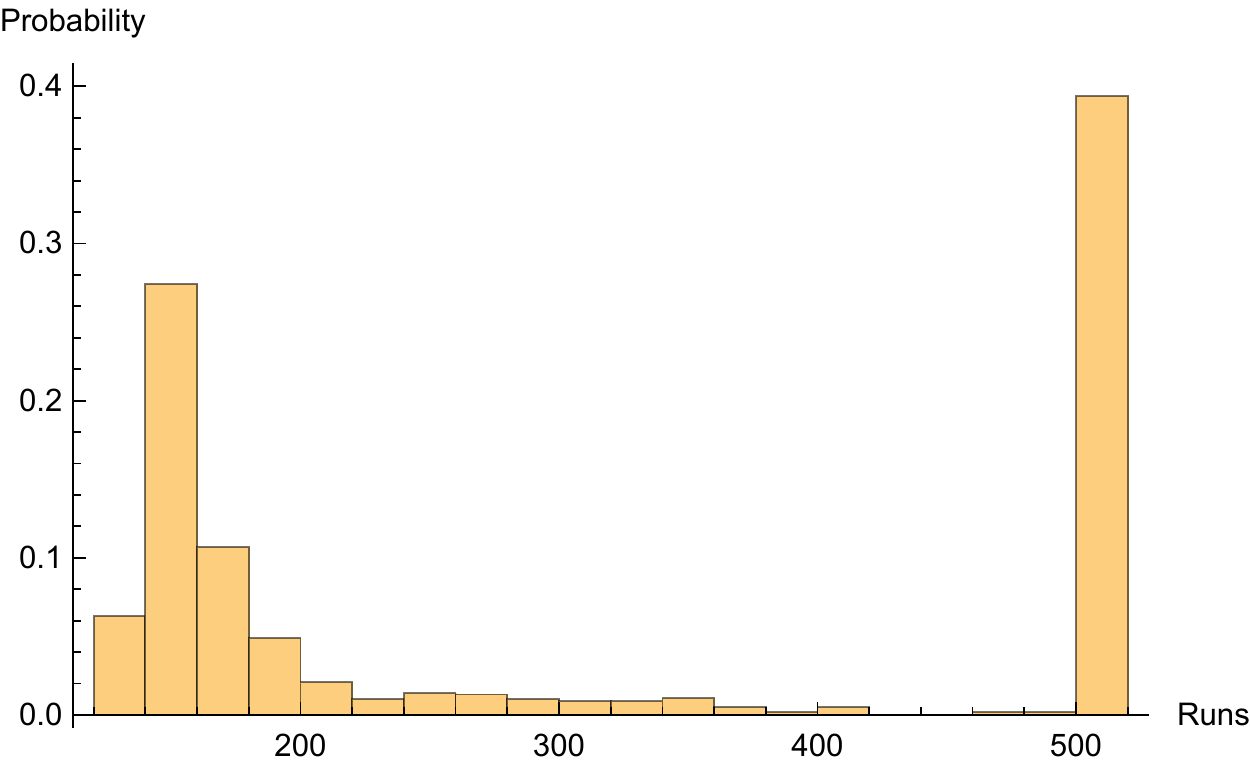}
	\caption{Numerical results for 1000 runs for each strength parameter. Left panel: probability of success (algorithm converges before 500 iterations) for each value of the strength parameter $\sigma$. Right panel: number of iterations for $\sigma_* = 0.6$.}
	\label{fig:results-norvig}
\end{figure}

Note that our algorithm for solving SudoQ does not succeed for every instance of the random initialization. This is a common feature of this type of algorithm based on alternative normalization, see e.g.~\cite[Table 3.1]{chi2012techniques}. We notice that the success probability depends on the strength parameter $\sigma$, and we choose (empirically) the optimal value, here $\sigma_* = 0.6$, for which the probability that the algorithm converges before ${I_{\max}} = 500$ is $60.6\%$. For this particular value, we plot the distribution of the number of iterations the algorithm ran for. Conditioning on the algorithm finishing before ${I_{\max}}$ steps, we notice a peak in the small values (approx 100). This suggests that if the algorithm does not converge quickly, it will not converge (meaning it will be stuck in some kind of local minimum of the loos function $\textsc{SinkhornError}$); the same phenomenon has been seen for other values of $\sigma$. Finally, let us point out that we have not found non-classical solutions to this grid, providing evidence for Conjecture \ref{conj:classical-1-quantum-1}. 

\bigskip

Further evidence for Conjecture \ref{conj:classical-1-quantum-1} is provided by the following exhaustive analysis of $2 \times 2$ grids, called \emph{Shi Doku} \cite{sudopedia}. For this simplified version of the game, it is easy to enumerate the hardest instances of the problem. 

\begin{definition}
    A classical grid $G \in CG_{n^2}$ is called \emph{minimal} if it admits a unique classical solution (i.e.~$G \in CUG_{n^2}$) and, moreover, all the grids obtained from $G$ by removing any of the clues lose this property (i.e.~they admit multiple classical solutions). 
\end{definition}

It turns out that there minimal number of clues of a uniquely solvable Shi Doku grid is 4, and that there are 13 such grids \cite{lass2012minimal}. We have ran our algorithm 100 times on each of these 13 grids, and found no counterexamples to Conjecture \ref{conj:classical-1-quantum-1}: whenever the algorithm terminated, the solution found was machine-precision-close to the (unique) classical one. We report the number of iterations (for the 1300 runs) in Figure \ref{fig:shidoku}. Out of the 1300 runs, Algorithm \ref{alg:SudoQ} did not terminate before 500 iterations on 6 instances, all corresponding to the following grid (we refer the reader to \href{https://github.com/inechita/SudoQ}{GitHub repository} accompanying the paper for the input grids and the output file): 
$$G_{05} = \begin{bmatrix}
	\cdot & \cdot & \cdot & \cdot \\
	\cdot & \cdot & \cdot & 1 \\
	\cdot & 2 & \cdot & \cdot \\
	3 & \cdot & 4 & \cdot
	\end{bmatrix}$$

\begin{figure}[ht]
	\centering
	\includegraphics[width=0.6\linewidth]{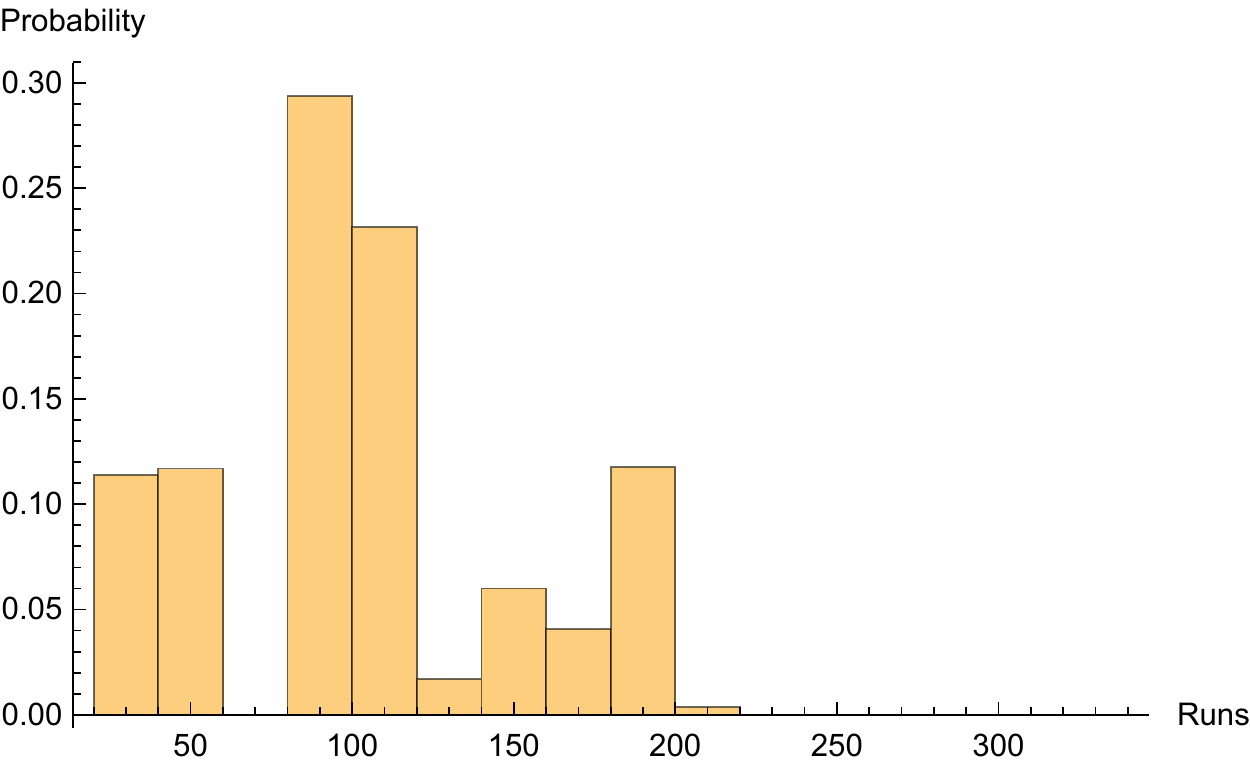}
	\caption{Numerical results for 100 runs of each of the 13 minimal Ski Doku ($2 \times 2$) grids with 4 clues.}
	\label{fig:shidoku}
\end{figure}

\section{Mixed SudoQ}\label{sec:mixed}

We consider in this section a more general framework, where we replace the unit-rank projections from Definition \ref{def:quantum-grid} with \emph{mixed quantum states} (i.e.~\emph{density matrices}). We recall that these are positive semidefinite operators of unit trace: $\rho \geq 0$ and $\Tr \rho =1$. This setting is a generalization of both :
\begin{itemize}
	\item the ``Random Sudoku'' grids from \cite[Section IV]{moon2009sinkhorn}, where the authors consider (classical) Sudoku grids where the entries are no longer perfectly determined, but might take one of several values with pre-assigned probabilities. Our generalization replaces thus probability vectors (or diagonal density matrices) by arbitrary density matrices. 
	\item the SudoQ grids and squares introduced in this paper in Section \ref{sec:grids}. Our generalization replaces pure quantum states (vectors) by mixed quantum states (density matrices)
\end{itemize}

In the context of the current work, there is a clear advantage in working with arbitrary density matrices rather than pure states. Pure states in our context can be represented by vectors on the unit sphere of the Hilbert space (up to a phase), whereas mixed states form a convex set. Given the linear structure that comes with a convex body, we can therefore expect that mixed states are easier to manipulate than pure states since we could apply the machinery of convex optimization.

It is interesting to consider to which extent (random) mixed states put in the empty cells of a classical grid \(X \in \mathrm{CG}_n \) are enforced to become pure states through Sinkhorn scaling procedure. This question is the quantum analog of the one answered in \cite{moon2009sinkhorn} and \cite{gunther_entropy_2012}. In these papers the authors define a ``constraint matrix'' which must be a permutation matrix when the constraints of the classical Sudoku grid are all satisfied. This permutation matrix is then approximated by a bistochastic matrix obtained via Sinkhorn scaling. Often the constraints imposed by the pre-filled entries automatically imply that the bistochastic matrix is effectively a permutation matrix at the end of the Sinkhorn scaling, but this is not a general fact. According to Birkhoff-von-Neumann theorem, permutation matrices are the extremal points of the convex polytope formed by bistochastic matrices. Therefore the authors are lead to introduce notions related to convex optimization in order to understand the convergence of Sinkhorn like algorithm for Sudoku. In our  context, the goal is almost the same, we would like to understand when (non-commutative) Sinkhorn scaling and quantum Sudoku constraints enforce an interior point of a convex set (in our case an arbitrary density matrix) to be moved to a pure state. We leave these questions open for future work.

\section{Application: quantum Sudoku code}\label{sec:application}

The \emph{Sudoku code} is a non-linear error correcting code used mainly for erasure channels. In \cite{moon_multiple_2006} the authors argue that a Sinkhorn like algorithm can be used in the problem of decoding Sudoku code, their method being strongly related to Bayesian belief propagation.
In this section, after recalling some definitions on quantum erasure channels, we introduce a quantum version of Sudoku code. Then we highlight the fact that the algorithm described in Section \ref{sec:sudoq-algorithm} can be used in the decoding step.

First, we consider the following quantum alphabet: \( \mathcal{X}_{n}=\{\ket{1}, \ket{2},...,\ket{n^2} \} \), which are just the kets of the canonical basis of \( \mathbb{C}^{n^2} \). We introduce another quantum alphabet \( \mathcal{Y}_n = \{ \mathcal{X}_{n}, \ket{x} \} \), the vector \( \ket{x} \in \mathbb{C}^{n^2+1} \) being orthogonal to \( \mathbb{C}^{n^2} \).
\begin{definition}
    A quantum erasure channel is a completely positive trace preserving map given by:
    \begin{equation}
        \ketbra{\phi}{\phi} \mapsto (1-\epsilon)\ketbra{\phi}{\phi}+\epsilon \ketbra{x}{x},
    \end{equation}
    where \(\ket{\phi} \in \mathcal{X}_{n} \) and its image belongs to \( \mathcal{Y}_n \subset \mathbb{C}^{n^2} \oplus \mathbb{C}\ket{x} = \mathbb{C}^{n^2+1} \).
\end{definition}
In this definition the vector \( \ket{x} \) represents the lost of a symbol belonging to \( \mathcal{X}_{n} \) (it is the quantum analog of the symbol ``?'' for a classical erasure channel) \cite{delfosse_upper_2012}.\\
We will now consider an information sequence of symbols belonging to \( \mathcal{X}_{n} \). This sequence is mapped to an element \(X \) of \( \CS_{n^2} \), which means we think about the set \( \CS_{n^2} \) as a code-words alphabet. Therefore the goal of such a procedure is to encode the redundant information into the Sudoku constraints. In order to draw some comparisons with other error correcting codes (such as Gallager code) it is usual to represent Sudoku code via a bipartite graph called Tanner graph. This graph contains two sets of nodes, the first set corresponds to Sudoku constraints given by Example \ref{ex:Sudoku-constraints}, whereas the second set is simply all the one dimensional projectors elements of \(X \), a projector node is linked to a constraint node when the projector is submitted to the corresponding constraint.
\begin{example}
Tanner graph for an element of $\CS_{4}$ : \\
\begin{tikzpicture}[thick,
  fsnode/.style={draw,square},
  fsnode/.style={fill=myblue},
  ssnode/.style={draw,circle},
  ssnode/.style={fill=mygreen},
  -,shorten >= 3pt,shorten <= 3pt
]

\begin{scope}[start chain=going right,node distance=7mm]
\foreach \i in {1,2,...,12}
  \node[fsnode,on chain] (f\i) [label=above: $C_{\i}$] {};
\end{scope}

\begin{scope}[xshift=-1cm,yshift=-4cm,start chain=going right,node distance=7mm]
\foreach \i in {1,2,...,4}
\foreach \j in {1,2,...,4}
  \node[ssnode,on chain] (s\i \j) [label=below: $P_{\i \j}$] {};
\end{scope}

\node [myblue,fit=(f1) (f12)] {};
\node [mygreen,fit=(s11) (s44)] {};

\draw (f1) -- (s11);
\draw (f1) -- (s12);
\draw (f1) -- (s13);
\draw (f1) -- (s14);

\draw (f2) -- (s21);
\draw (f2) -- (s22);
\draw (f2) -- (s23);
\draw (f2) -- (s24);

\draw (f3) -- (s31);
\draw (f3) -- (s32);
\draw (f3) -- (s33);
\draw (f3) -- (s34);

\draw (f4) -- (s41);
\draw (f4) -- (s42);
\draw (f4) -- (s43);
\draw (f4) -- (s44);

\draw (f5) -- (s11);
\draw (f5) -- (s21);
\draw (f5) -- (s31);
\draw (f5) -- (s41);

\draw (f6) -- (s12);
\draw (f6) -- (s22);
\draw (f6) -- (s32);
\draw (f6) -- (s42);

\draw (f7) -- (s23);
\draw (f7) -- (s13);
\draw (f7) -- (s33);
\draw (f7) -- (s43);

\draw (f8) -- (s14);
\draw (f8) -- (s24);
\draw (f8) -- (s34);
\draw (f8) -- (s44);

\draw (f9) -- (s11);
\draw (f9) -- (s12);
\draw (f9) -- (s21);
\draw (f9) -- (s22);

\draw (f10) -- (s13);
\draw (f10) -- (s14);
\draw (f10) -- (s23);
\draw (f10) -- (s24);

\draw (f11) -- (s31);
\draw (f11) -- (s32);
\draw (f11) -- (s41);
\draw (f11) -- (s42);

\draw (f12) -- (s33);
\draw (f12) -- (s34);
\draw (f12) -- (s43);
\draw (f12) -- (s44);
\end{tikzpicture}
\end{example}

Then the element \( X \) is sent through a noisy quantum channel modelised by a quantum erasure channel.  After the transmissions through the quantum erasure channel some symbols in the sequence are lost, but thanks to the Sudoku constraints and the remaining unaffected symbols we could be able to recover the initial sequence.
\begin{example}
Action of a quantum erasure channel on an element of $ \CS_4 $ : \\

$$\begin{array}{cccc}

\begin{bmatrix}
  	\ket{1} & \ket{2} & \ket{3} & \ket{4} \\
	\ket{3} & \ket{4} & \ket{2} & \ket{1} \\    
	\ket{4} & \ket{3} & \ket{1} & \ket{2} \\
	\ket{2} & \ket{1} & \ket{4} & \ket{3}
\end{bmatrix} \hspace{0.2\textwidth}&
\longrightarrow &
 \begin{bmatrix}
   	\ket{x} & \ket{2} & \ket{3} & \ket{x} \\
	\ket{3} & \ket{x} & \ket{x} & \ket{1} \\    
	\ket{4} & \ket{x} & \ket{x} & \ket{2} \\
	\ket{x} & \ket{1} & \ket{4} & \ket{x}
\end{bmatrix} 
\end{array}$$

\end{example}
In order to decode the transmitted sequence we can replace the ket \( \ket{x} \) by a random vector and apply the algorithm described in section 5. In this context we understand the importance of knowing when it is possible or not to uniquely recover a grid in \( \CS_{n^2} \) from an element of \( \CUG_{n^2} \) (i.e.~Conjecture \ref{conj:classical-1-quantum-more}). Indeed we need this result to characterize the performance of such a quantum Sudoku code, since this performance is related to the number of symbols which can be erased without affecting the initial message.

\section{Conclusion}
Sudoku puzzles are important because they constitute a simple but non-trivial example of SAT problems, and are proved to be NP-complete problems. In this paper we have shown how this particular SAT problem translates in the quantum world, and how we can solve it thanks to a non-commutative generalization of Sinkhorn algorithm. We also underlined the fact that these new definitions allow us to describe a new non-linear quantum error correcting code for quantum erasure channels. The notions introduced here require further investigations, in particular to give an answer to the two conjectures from Section \ref{sec:sudoq-solutions}.  

\bigskip

\noindent\textit{Acknowledgements.} The authors would like to thank the organizers of the program \emph{Operator Algebras, Groups and Applications to Quantum Information} held at the ICMAT (Madrid) in 2019 for the support and for providing the opportunity for this research to be started. Most of this research was done during J.P.'s internship at the LPT Toulouse, which we would also like to thank for support. 

\bibliographystyle{alpha}
\bibliography{sudoq}

\end{document}